\documentclass[12pt]{article}
\usepackage[margin=1.25in]{geometry}
\usepackage{amsmath,amssymb,amsthm,mathrsfs}
\usepackage{enumitem}
\usepackage{setspace}
\usepackage[authoryear]{natbib}
\usepackage{bbm}
\usepackage[colorlinks,citecolor=blue,linkcolor=blue,urlcolor=blue]{hyperref}
\allowdisplaybreaks
\theoremstyle{plain}
\newtheorem{thm}{Theorem}
\newtheorem{lem}{Lemma}
\newtheorem{prop}{Proposition}
\newtheorem{cor}{Corollary}
\theoremstyle{definition}
\newtheorem{defn}{Definition}
\newtheorem{example}{Example}
\newtheorem{remark}{Remark}

\newcommand{\E}{\mathbb{E}}
\newcommand{\R}{\mathbb{R}}
\newcommand{\PP}{\mathbb{P}}
\newcommand{\cU}{\mathcal{U}}
\newcommand{\cX}{\mathcal{X}}
\newcommand{\ind}{\mathbbm{1}}

\newcommand{\norm}[1]{\lVert #1 \rVert}

\title{When ``Normalization Without Loss of Generality'' Loses Generality\footnote{I thank Xiaohong Chen, Jiafeng Chen, Phil Haile, Soonwoo Kwon, Whitney Newey, Peter Phillips, and Elie Tamer for helpful discussions and suggestions. Part of this paper builds on an unpublished appendix to my working paper \cite{gao2017arXiv} on network formation, the main results of which were published as \citet{gao2020nonparametric}. The present paper develops  a stand-alone framework and substantially extends the analysis. AI tools (Claude, ChatGPT, Gemini, and Refine.ink) were used for research assistance and critical reviews; the author assumes full responsibility for any remaining errors.}}
\author{Wayne Yuan Gao\thanks{Gao: Department of Economics, University of Pennsylvania, 133 South 36th Street, Philadelphia, PA 19104, waynegao@upenn.edu.}}
\begin{document}
	\maketitle
	
	\begin{abstract}
		\noindent Normalization is ubiquitous in economics, and a growing literature shows that ``normalizations'' can matter for interpretation, counterfactual analysis, misspecification, and inference. This paper provides a general framework for these issues, based on the formalized notion of \emph{modeling equivalence} that partitions the space of unknowns into equivalence classes, and defines normalization as a WLOG selection of one representative from each class. A counterfactual parameter is normalization-free if and only if it is constant on equivalence classes; otherwise any point identification is created by the normalization rather than by the model. Applications to discrete choice, demand estimation, and network formation illustrate the insights made explicit through this criterion. We then study two further sources of fragility: an extension trilemma establishes that fidelity, invariance, and regularity cannot simultaneously hold at a boundary singularity, while a normalization can itself introduce a coordinate singularity that distorts the topological and metric structures of the parameter space, with consequences for estimation and inference.
		
		\medskip
		\noindent\textbf{Keywords:} identification, normalization, without loss of generality, equivalence class, counterfactual, discrete choice
		
		%(JEL C14, C21, C25, C51)
	\end{abstract}
	
	\bigskip
	
	\section{Introduction}
	\label{sec:intro}
	
	Normalization is standard practice in econometrics and applied economics. In binary choice models, the error variance or the norm of a structural parameter is set to unity. In demand estimation, the idiosyncratic taste shock follows a Type-I extreme value distribution, with the latent utility of the outside option often set to zero. In network formation models, parameter norms or quantiles of the unobserved shock are fixed. In entry games, the variance of the private shock is set to one.
	In each case, certain unknowns, typically the location and scale of an unobservable, are not separately identified from the structural parameters, and a normalization pins them down.
	
	Normalization is not merely a notational convenience. A substantial literature now shows, in different ways, that restrictions described as normalizations can affect interpretation, counterfactuals, misspecification, and the geometry of inference. Point identification results are still routinely stated ``under normalization,'' and the implications of that qualifier for counterfactual analysis are often left implicit. This creates two recurring problems for applied work.
	First, the distinction between ``identified'' and ``identified under normalization'' is easily blurred: once a normalization has been imposed, the resulting parameter values are treated as uniquely identified structural objects, when in fact they are only one labeling of an equivalence class. Second, counterfactual parameters such as welfare changes or percentage welfare changes may or may not be invariant to the normalization, and when they are not, any apparent point identification is normalization-relative rather than model-driven. The difficulty is not confined to one class of models, but it is usually diagnosed case by case.
	
	An important clarification of scope is in order. This paper does not take a stand on which models, maintained assumptions, or objects of interest a researcher should adopt. Researchers are free to impose whatever substantive assumptions their application demands, and to defend those assumptions to their respective audiences. Our point is narrower and conditional: \emph{given} a researcher's chosen model---with all its structural equations and maintained assumptions held fixed---the paper provides a framework for diagnosing normalization issues and their downstream consequences explicitly. The ``safest'' counterfactuals are those written directly on observable quantities (choice probabilities, market shares, linking probabilities), which are automatically normalization-free. When this is not feasible, as in many structural settings where the objects of interest involve latent quantities, the researcher should at minimum verify that each target counterfactual is measurable with respect to the equivalence class induced by the model's normalization, and make the equivalence class structure explicit.

	Percentage changes can be particularly dangerous, because when the normalization involves a location shift, the denominator of the percentage is itself normalization-dependent. A simple daily-life analogy makes this concrete: consider a warming of $10^\circ$C in physical temperature from $1^\circ$C to $11^\circ$C across two days, and the same physical temperature change simultaneously corresponds to a $1{,}000\%$ increase in Celsius, a $53\%$ increase in Fahrenheit, but only a $3.6\%$ increase in Kelvin. The percentage is an artifact of the zero point, not of the phenomenon, and a report of ``$1{,}000\%$ increase in temperature'' without proper disclosure of the underlying unit (and the affine transformation that defines the unit) can be misleading. Just as Celsius, Fahrenheit, and Kelvin are related by positive affine transformations, the indeterminacy in discrete response models is often generated (and then normalized away) by transformations of the same form. Whenever a structural model identifies utility or surplus only up to a location-scale transformation, raw percentage welfare changes inherit the arbitrariness of the normalization's zero point.
	
	A striking recent illustration of a related issue appears in \citet{chen2024logs}, who show that for any ``log-like'' transformation such as $\log(1+Y)$ or $\operatorname{arcsinh}(Y)$, the ATE is \emph{arbitrarily} sensitive to the units of $Y$ whenever the treatment moves individuals onto or off the boundary $Y=0$. In the language of the present paper, the percentage effect has a singularity at $Y=0$, and any completion that ``patches'' this singularity imports an auxiliary convention (units) rather than a well-defined property of the model. Section~\ref{subsec:augmentation} formalizes this insight as a general extension trilemma.
	
	\medskip
	
	Against this backdrop, the paper develops a general framework built on three interlocking concepts.
	\emph{Modeling equivalence} is a bijective transformation of the unknowns that preserves the structural equation \emph{and} all maintained assumptions. The notion of modeling equivalence is strictly stronger than classical observational equivalence, which allows arbitrary changes in the model so long as the observable data distribution is preserved.
	We then use modeling equivalence to define a partition of the space of unknowns in a model into equivalence classes. Proper \emph{normalization} selects one representative from each class: it is a labeling device \emph{without loss of generality}, not a source of identifying information. 
	
	We then articulate how normalization, even when imposed without loss of generality, affects the identifiability of a \emph{counterfactual parameter}, defined as any known function of the unknowns. Our main results are a lemma on the \emph{limit of identifiability}, showing that any counterfactual that varies within some equivalence class cannot be point-identified for all true parameter values, and a theorem on \emph{WLOG normalization}, establishing that both the identifiability and the identified value of any counterfactual are invariant to the choice of normalization. Together, these provide a clean necessary condition: for a counterfactual to be point-identified, it must be constant on every modeling equivalence class. In other words, it must factor through the quotient space induced by the model's indeterminacy. Crucially, verifying this congruence condition requires an explicit statement of the modeling equivalence relation, which is encoded in the original unnormalized model but cannot be recovered from the normalized model alone.
	
	We then turn to the interaction between normalization and singularity. An \emph{extension trilemma} result shows that when a counterfactual is normalization-free on a regular subdomain but the induced quotient functional has a singularity at the boundary, no completion can simultaneously preserve fidelity, invariance, and regularity. This provides a complementary perspective to the impossibility result in \citet{chen2024logs}, framed in terms of arbitrary equivalence structures and quotient topologies: extensions that ``patch'' boundary singularities are not innocuous normalizations; they necessarily change the object of interest.
	The singularity story has a flip side: whereas domain augmentation attempts to repair a pre-existing singularity, a normalization can \emph{create} one. Dividing by a distinguished coefficient maps the excluded hyperplane to infinity, introducing a coordinate singularity that distorts the topology and metric of the parameter space. Using binary choice scale normalization as a case study, we show that the special-coordinate normalization destroys compactness and connectedness and can reverse convergence verdicts, while the sphere normalization preserves these properties by construction. More broadly, whether a normalization is WLOG cannot be assessed in isolation, but must be assessed end-to-end, jointly with the maintained model, the target counterfactual, and any auxiliary structure used downstream.
	
	We illustrate the framework through three applied settings: binary choice, multinomial choice (BLP-type demand models), and dyadic network formation, showing that marginal effects, market shares, and linking probabilities are normalization-free, while welfare levels and percentage welfare changes are generically normalization-dependent.
	Money-metric welfare \emph{changes} occupy a middle ground: in the standard quasilinear logit model, the location shift cancels in the difference and the scale cancels against the price coefficient, making $\Delta CS$ normalization-free. That said, this cancellation is model-specific and does not extend to welfare levels or percentages. When normalizations or other conventional assumptions are adopted, we argue that the normalization choice, the underlying equivalence class, and the dependence of each target counterfactual on that choice should all be made explicit, especially to applied researchers and policy makers who rely on the reported quantities.
	
	\subsubsection*{Literature Review}
	
	The literature most closely related to ours comes from three partly separate strands.
	An important starting point is \citet{lewbel2019identification}, especially Section~6.3, which highlights that whether a restriction is genuinely without loss depends on the object of interest, and that normalization choice can matter for interpretation and precision even when it is innocuous for point identification.
	Our paper formalizes and extends these insights by introducing modeling equivalence as a refinement of observational equivalence, defining the quotient space of equivalence classes, proving WLOG as a theorem rather than asserting it informally, and giving a single general invariance criterion for normalization-free counterfactuals.
	
	The second strand studies normalization as a problem of parameter-space geometry and inference.
	\citet*{hamilton2007normalization} show that poor normalizations can distort the geometry of the parameter space, producing multimodal posteriors, disjoint confidence regions, and misleading uncertainty statements in likelihood-based models.
	Their ``identification principle'' says that the boundary of a normalized parameter space should line up with loci of local nonidentification.
	Our analysis in Section~\ref{sec:sphere} is closely related in spirit, but differs in scope and object: rather than start from a finite-dimensional likelihood, we work with modeling-equivalence classes in the space of unknowns and ask which counterfactuals, regularity conditions, and inferential objects descend to the quotient.
	
	The third strand shows in concrete structural models that some restrictions described as normalizations are in fact substantive assumptions.
	\citet{agostinelli2025estimating} show that period-by-period re-normalization in dynamic latent-factor models can be over-identifying and restrictive when the production technology already has known location and scale, and can bias estimation; their earlier working paper \citep{agostinelli2016identification} develops the normalization analysis in detail.
	\citet{freyberger2025normalizations} likewise shows in skill-formation models that seemingly innocuous scale and location restrictions may constrain identified parameters and alter counterfactuals, depending on the production technology, measurement system, and estimation method.
	\citeauthor{freyberger2025normalizations} also provides a formal definition of normalization: a restriction $\tau \in \mathcal{T}_N$ is a normalization with respect to $g(\tau_0)$ if $\{g(\tau) : \tau \in \mathcal{T}_0 \cap \mathcal{T}_N\} = \{g(\tau) : \tau \in \mathcal{T}_0\}$, which is restriction-based and function-specific.
	Our contribution is to provide the general language that makes these findings portable across models: a restriction is a true normalization only if it amounts to choosing a representative within a modeling-equivalence class, while failures arise either because the target counterfactual is not invariant on the class or because the downstream object does not extend coherently to the quotient. In addition, our notion of modeling equivalence is strictly stronger than the observational equivalence used in all three strands, because it requires preservation of all maintained assumptions, not just the data distribution (see Section~\ref{subsec:me_oe}).
	
	Normalization has also been discussed in many specific settings: semiparametric binary response \citep{manski1988identification}, panel binary choice \citep{chamberlain2010binary}, demand estimation \citep{berry1995automobile, berry2014identification}, entry games \citep{tamer2003incomplete, ciliberto2009market}, network formation \citep{graham2017econometric, gao2020nonparametric}, and dynamic discrete choice \citep{rust1994structural, magnac2002identifying}, where the outside-option normalization is famously \emph{not} free.
	
	A related and growing literature studies the identification of counterfactual parameters in structural models.
	\citet*{kalouptsidi2021identification} characterize which counterfactuals are identified in dynamic discrete choice models despite non-identification of the full model, and note that scale normalizations cease to be innocuous once cardinal objects such as value functions enter the counterfactual.
	\citet*{bhattacharya2023demand} show that in standard discrete choice the outside-option normalization is innocuous for welfare analysis, whereas with social interactions it is not, because the outside-option utility itself becomes policy-dependent.
	\citet*{gu2025counterfactual} develop a general framework for counterfactual identification in discrete outcome models.
	These papers study what can be identified \emph{after} a normalization has been imposed, and some point out in model-specific terms that a purported normalization is or is not innocuous.
	Our contribution is different: we isolate normalization itself as a general, cross-model source of counterfactual fragility, formalized through modeling-equivalence classes, quotient factorization, interior non-invariance, and boundary or chart singularities.
	
	\medskip
	
	The rest of the paper is organized as follows.
	Section~\ref{sec:framework} presents the general framework. Section~\ref{sec:examples} develops three applied examples. Section~\ref{sec:modification} examines how normalization interacts with singularities, both in the target functional and in the normalized parameter space. Section~\ref{sec:discussion} introduces end-to-end WLOG normalization and offers practical recommendations. Section~\ref{sec:conclusion} concludes.
	
	%% ================================================================
	\section{General Framework}
	\label{sec:framework}
	%% ================================================================
	
	\subsection{An Econometric Model}
	\label{subsec:model}
	
	We write a generic econometric model, i.e., a collection of structural equations together with all maintained assumptions, abstractly as
	\begin{equation}\label{eq:model}
		f(\mathbf{X}, \mathbf{U}, \gamma_0) = 0,
	\end{equation}
	where $\mathbf{X}$ collects all observable random elements (data), $\mathbf{U} \in \cU$ collects all unobservable random elements (errors and latent variables), $\gamma_0 \in \Gamma$ is the true value of all unknown fixed parameters (possibly infinite-dimensional), and $f$ is a known functional encoding the structural equation \emph{and} all maintained assumptions, such as distributional assumptions, independence restrictions, support conditions, and any other maintained hypotheses.  We emphasize that $f$ incorporates every aspect of the model: there is no fundamental distinction between ``model specification'' and ``model assumption.'' Throughout this paper, we take $f$ as given---the choice of model, including all distributional assumptions, independence restrictions, and functional-form specifications, is the researcher's prerogative and is not the object of our analysis. Our analysis begins \emph{after} the researcher has settled on $f$ and asks: what equivalence classes does this model induce, and how does normalization interact with downstream objects of interest?

	The tuple $(\mathbf{X}, \mathbf{U}, \gamma_0, f)$ provides a completely general decomposition of an econometric model into its observable, unobservable, unknown, and known components.
	We call $\cU \times \Gamma$ the \emph{space of unknowns}.%
	\footnote{More abstractly, let $\Theta$ denote the space of unknowns and let $\mathcal{L} : \Theta \rightrightarrows \mathscr{P}(\cX)$ be the observable implication map, where $\mathcal{L}(\theta)$ collects all observable laws consistent with the model when the state equals $\theta$. The structural representation $f(\mathbf{X}, \mathbf{U}, \gamma_0) = 0$ is one way to generate $\mathcal{L}$; the definitions below apply to any~$\mathcal{L}$, and in particular accommodate incomplete models and models with nuisance variation.}
	
	\subsection{Modeling Equivalence}
	\label{subsec:me}
	
	\begin{defn}[Modeling Equivalent Transformation]\label{def:met}
		A bijection $\psi : \cU \times \Gamma \to \cU \times \Gamma$ is a \emph{modeling equivalent transformation} if
		\begin{equation}\label{eq:met}
			f\bigl(\mathbf{X}, \psi(\mathbf{U}, \gamma)\bigr) = f(\mathbf{X}, \mathbf{U}, \gamma), \quad \forall\, (\mathbf{U}, \gamma) \in \cU \times \Gamma.
		\end{equation}
	\end{defn}
	
	In other words, $\psi$ preserves the entire model $f$, not merely the distribution of the data, but all maintained assumptions as well. 
	
	\begin{defn}[Modeling Equivalence]\label{def:me}
		Two points $(\mathbf{U}, \gamma)$ and $(\mathbf{U}', \gamma')$ in $\cU \times \Gamma$ are \emph{modeling equivalent}, written $(\mathbf{U}, \gamma) \sim (\mathbf{U}', \gamma')$, if there exists a modeling equivalent transformation $\psi$ such that $(\mathbf{U}', \gamma') = \psi(\mathbf{U}, \gamma)$.
	\end{defn}
	
	Note that modeling equivalence is strictly stronger than the classical notion of \emph{observational equivalence}, which requires only that two configurations of the unknowns generate the same data distribution. Observational equivalence allows a simultaneous change in the maintained assumptions (e.g., switching distributional families), while modeling equivalence does not. See Section~\ref{subsec:me_oe} for more detail about this distinction.
	
	\medskip
	
	Modeling equivalence is a well-defined equivalence relation on $\cU \times \Gamma$: the identity map establishes reflexivity, the inverse of a bijection satisfying \eqref{eq:met} again satisfies \eqref{eq:met} (symmetry), and the composition of two such bijections satisfies \eqref{eq:met} (transitivity).
	
	The resulting quotient space is thus a partition, which we call the \emph{modeling equivalence partition}, defined by
	\begin{equation}\label{eq:ME_partition}
		Q_\sim := (\cU \times \Gamma) / {\sim}.    
	\end{equation}
	We write $[\mathbf{U}, \gamma]$ for the equivalence class containing $(\mathbf{U}, \gamma)$.
	Throughout, $Q_\sim$ carries the quotient topology induced by the canonical projection $\pi : \cU \times \Gamma \to Q_\sim$, so that $V \subseteq Q_\sim$ is open if and only if $\pi^{-1}(V)$ is open in $\cU \times \Gamma$.
	
	\subsection{Modeling Equivalence versus Observational Equivalence}
	\label{subsec:me_oe}
	
	The standard notion of observational equivalence requires only that two configurations generate the same data distribution: $(\mathbf{U}, \gamma) \sim_{OE} (\mathbf{U}', \gamma')$ if the observable distributions of data they induce are identical.
	Clearly, modeling equivalence implies observational equivalence, since \eqref{eq:met} ensures the data distribution is preserved, but the converse is in general false, because observationally equivalent transformations may violate maintained assumptions.
	This distinction matters because normalization constrains the unknowns \emph{within} the model's maintained assumptions, but it should not change the model itself. The relevant indeterminacy for normalization analysis is therefore characterized by modeling equivalence, not by observational equivalence.
	
	The following example illustrates this difference.

	\begin{example}[Fixed Effects in Dyadic Network formation]\label{ex:oe_not_me_network}
		Consider the dyadic network formation model $D_{ij} = \ind\{w(X_i, X_j) + A_i + A_j \ge U_{ij}\}$ under cross-sectional random sampling and i.i.d. pairwise shocks $U_{ij}$ \`a la \cite{graham2017econometric} and \cite{gao2020nonparametric}, which is analyzed in more detail in Section \ref{subsec:network}. One may define a transformation of the unknown error and fixed effects by  $\hat{A}_i := A_i - A_1$ and $\hat{U}_{ij} := U_{ij} - 2A_1$, i.e., ``normalize'' with respect to individual 1's fixed effect $A_1$. Since $w(X_i, X_j) + \hat{A}_i + \hat{A}_j - \hat{U}_{ij} = w(X_i, X_j) + A_i + A_j - U_{ij}$, the transformed specification produces exactly the same data distribution; hence, the transformation preserves observational equivalence.  However, while $A_i$ is iid across $i$ under the assumption of cross-sectional random sampling, $\hat{A}_i = A_i - A_1$ is clearly \emph{not} iid (all $\hat{A}_i$ are correlated through $A_1$, and furthermore $\hat{A}_i=0$ is degenerate), violating the assumption of cross-sectional random sampling. Hence this transformation does \emph{not} preserve modeling equivalence.
	\end{example}
	
	This example shows that the set of modeling equivalent transformations is a proper subset of the set of observational-equivalence-preserving transformations.
	Using observational equivalence for normalization analysis would overstate the indeterminacy by including transformations that are inadmissible given the model's maintained assumptions.
	
	\subsection{Normalization}
	\label{subsec:norm}
	
	\begin{defn}[Normalization]\label{def:norm}
		A mapping $\psi_N : \cU \times \Gamma \to \cU \times \Gamma$ is a \emph{normalization} if both the following hold:
		\begin{enumerate}[label=(\roman*), nosep]
			\item Within-class collapse: $(\mathbf{U}, \gamma) \sim (\mathbf{U}', \gamma') \implies \psi_N(\mathbf{U}, \gamma) = \psi_N(\mathbf{U}', \gamma') \sim (\mathbf{U}, \gamma)$.
			\item Across-class separation: $(\mathbf{U}, \gamma) \not\sim (\mathbf{U}', \gamma') \implies \psi_N(\mathbf{U}, \gamma) \neq \psi_N(\mathbf{U}', \gamma')$.
		\end{enumerate}
	\end{defn}
	
	In other words, a normalization selects exactly one representative from each equivalence class. It is not bijective by design: all elements within a given class are mapped to a single element. Different global normalizations used for binary choice models, for example, $\varepsilon_i \sim \mathcal{N}(0,1)$ versus $\norm{\boldsymbol{\beta}_0} = 1$, are simply different labelings of the same equivalence classes. By contrast, a coordinate restriction such as $|\beta_{0,1}| = 1$ is best viewed as a \emph{local chart} on the two open hemispheres where $\beta_{0,1} \neq 0$, not as a globally available normalization; it breaks down as $\beta_{0,1} \to 0$ (see Section~\ref{sec:sphere}).
	
	\begin{prop}[Normalization Induces a Bijection]\label{prop:bijection}
		Every normalization $\psi_N$ induces a bijection between $Q_\sim$ and the normalized parameter space $\psi_N(\cU \times \Gamma)$.
	\end{prop}
	
	\begin{proof}
		Define $s_N : Q_\sim \to \psi_N(\cU \times \Gamma)$ by $s_N([\mathbf{U}, \gamma]) = \psi_N(\mathbf{U}, \gamma)$.
		This is well-defined by condition~(i) of Definition~\ref{def:norm}.
		Surjectivity is by construction; injectivity follows from condition~(ii).
	\end{proof}
	
	Proposition~\ref{prop:bijection} shows that the normalized parameter space is simply a relabeling of the quotient: every normalization picks a set of representatives that is in one-to-one correspondence with the equivalence classes. The choice among normalizations is therefore a choice among coordinate systems for the same underlying object.%
	\footnote{\citet{freyberger2025normalizations} provides a complementary definition: a restriction $\tau \in \mathcal{T}_N$ is a normalization with respect to $g(\tau_0)$ if $\{g(\tau) : \tau \in \mathcal{T}_0 \cap \mathcal{T}_N\} = \{g(\tau) : \tau \in \mathcal{T}_0\}$. This is restriction-based and function-specific. Our definition is structure-based: the equivalence class partition governs all functions simultaneously, and a function is normalization-free if and only if it factors through the quotient $Q_\sim$.}
	
	\subsection{Identifiability of Counterfactual Parameters}
	\label{subsec:cf}
	
	\begin{defn}[Counterfactual Parameter]\label{def:cf}
		A \emph{counterfactual} is a known function $q : \cU \times \Gamma \to \R$.
		The \emph{counterfactual parameter} is the value $q(\mathbf{U}_0, \gamma_0)$ evaluated at the true unknowns $(\mathbf{U}_0, \gamma_0)$.
	\end{defn}
	
	Counterfactual parameters are the ultimate objects of economic interest. They encode the answers to questions such as ``what is the marginal effect of $x_j$?'' or ``what is the welfare change from a policy intervention?''
	
	Technically, a counterfactual is a known \emph{functional} defined on the space of the unknowns. We use the name ``counterfactual'' to make explicit its standard use in applied economics. 
	
	\begin{lem}[Limit of Identifiability of A Counterfactual]\label{lem:limit}
		If a counterfactual $q : \cU \times \Gamma \to \R$ is not constant on some modeling equivalence class, then there exist true values $(\mathbf{U}_0, \gamma_0)$ such that $q(\mathbf{U}_0, \gamma_0)$ cannot be point-identified.
	\end{lem}
	
	\begin{proof}
		Since $q$ is not $Q_\sim$-measurable, there exists a modeling equivalence class $p \in Q_\sim$ and two elements $(\mathbf{U}, \gamma), (\mathbf{U}', \gamma') \in p$ such that $q(\mathbf{U}, \gamma) \neq q(\mathbf{U}', \gamma')$.
		
		Suppose the true values satisfy $(\mathbf{U}_0, \gamma_0) \in p$.
		Since $(\mathbf{U}_0, \gamma_0) \sim (\mathbf{U}, \gamma) \sim (\mathbf{U}', \gamma')$, there exist modeling equivalent transformations $\psi$~and~$\psi'$ with $(\mathbf{U}_0, \gamma_0) = \psi(\mathbf{U}, \gamma) = \psi'(\mathbf{U}', \gamma')$.
		Define $\bar{f}(\mathbf{X}, \cdot) := f(\mathbf{X}, \psi(\cdot))$ and $\underline{f}(\mathbf{X}, \cdot) := f(\mathbf{X}, \psi'(\cdot))$.
		By~\eqref{eq:met}, $\bar{f} \equiv f \equiv \underline{f}$.
		Under $\bar{f}$ the ``true value'' is $\psi^{-1}(\mathbf{U}_0, \gamma_0) = (\mathbf{U}, \gamma)$; under $\underline{f}$ it is $(\psi')^{-1}(\mathbf{U}_0, \gamma_0) = (\mathbf{U}', \gamma')$.
		Since $\bar{f} \equiv \underline{f} \equiv f$ and the observable data are the same, any valid identified set must contain both $q(\mathbf{U}, \gamma)$ and $q(\mathbf{U}', \gamma')$.
		Since these differ, $q(\mathbf{U}_0, \gamma_0)$ is not point-identified.
	\end{proof}
	
	The intuition is straightforward: if two modeling-equivalent configurations $(\mathbf{U}, \gamma)$ and $(\mathbf{U}', \gamma')$ have $q(\mathbf{U}, \gamma) \neq q(\mathbf{U}', \gamma')$, then the model, which by construction cannot distinguish these two configurations, must include both values in any  identified set.
	
	This suggests a useful classification of counterfactual statements. We say that $q$ is \emph{normalization-free} if it is constant on every equivalence class, and \emph{normalization-dependent} otherwise.
	Lemma~\ref{lem:limit} says that being normalization-free is \emph{necessary} for point identification: a normalization-dependent counterfactual is provably non-point-identified for some true parameter values.
	Being normalization-free is not by itself \emph{sufficient} for point identification, since there may be other, non-normalization sources of non-identification (which is model specific and not the focus of this paper). The paper's message is not about what \emph{can} be identified after normalization, but about what \emph{cannot} be identified regardless of normalization.
	
	\begin{cor}[Necessary Condition for Point Identification]\label{cor:char}
		If a counterfactual $q$ is point-identified for all true values $(\mathbf{U}_0, \gamma_0)$, then $q$ must be constant on every modeling equivalence class.
	\end{cor}
	
	The corollary is simply the contrapositive of Lemma~\ref{lem:limit}.
	Its content can be restated as a factorization condition: $q$ is normalization-free if and only if there exists a map $\bar{q} : Q_\sim \to \R$ such that
	\begin{equation}\label{eq:factorization}
		q = \bar{q} \circ \pi, \quad \text{for } \pi : \cU \times \Gamma \to Q_\sim,
	\end{equation}
	where $\pi$ is the \emph{canonical projection} that sends each $(\mathbf{U}, \gamma)$ to its equivalence class $[\mathbf{U}, \gamma]$.
	Informally, truly identified counterfactuals are functionals of equivalence classes rather than of arbitrary representatives.
	
	The lemma and corollary above allow us to separate several cases regarding the identification of a counterfactual parameter:
	\begin{enumerate}[nosep]
		\item \emph{Normalization-free point identification}: the counterfactual is constant on every equivalence class and is point-identified (the ideal case).
		\item \emph{Normalization-dependent ``point identification''}: the counterfactual varies within some equivalence class, yet is reported as point-identified after normalization.
		\item \emph{Point identification up to normalization} (or \emph{normalizable set identification}): the identified set, prior to normalization, is a union of equivalence classes that collapses to a point once a normalization is imposed.%
		\footnote{Normalizable set identification refers to a set identification result that admits a normalization with respect to the relevant modeling equivalence class, after which it becomes a point identification result up to the normalization.}
		\item \emph{Unnormalizable partial identification}: the identified set does not collapse to a point even after a proper normalization.%
		\footnote{Unnormalizable partial identification refers to a ``standard'' partial identification result that does not pin down a point even after a proper normalization.}
	\end{enumerate}
	Generically, we view case~(2) as an ill-posed notion of identification: a normalization-dependent counterfactual that is ``point-identified'' only under a specific and fundamentally arbitrary normalization. Such results should be interpreted as negative findings, or at the very least with great caution, since the reported value depends on a labeling convention rather than on the phenomenon itself.
	
	\medskip
	
	Once the subtlety above is clarified, a researcher can then make a justified WLOG claim about normalization as formalized in the following theorem. 
	
	\begin{thm}[Normalization WLOG]\label{thm:wlog}
		Let $q$ be any counterfactual and $\psi_N$ any normalization in the sense of Definition~\ref{def:norm}.
		Then the identified set of $q$ under the normalized model is identical to its identified set under the original model.
	\end{thm}
	
	\begin{proof}
		Let $\psi_N$ be any normalization in the sense of Definition~\ref{def:norm}.
		After imposing $\psi_N$, the model becomes $f_N(\mathbf{X}, \tilde{\mathbf{U}}, \tilde\gamma) := f(\mathbf{X}, \psi_N(\tilde{\mathbf{U}}, \tilde\gamma))$.
		By Definition~\ref{def:norm}(i), $\psi_N(\tilde{\mathbf{U}}, \tilde\gamma) \sim (\tilde{\mathbf{U}}, \tilde\gamma)$ for all $(\tilde{\mathbf{U}}, \tilde\gamma)$.
		By the definition of modeling equivalence \eqref{eq:met}:
		\begin{equation}
			f_N(\mathbf{X}, \tilde{\mathbf{U}}, \tilde\gamma) = f\bigl(\mathbf{X}, \psi_N(\tilde{\mathbf{U}}, \tilde\gamma)\bigr) = f(\mathbf{X}, \tilde{\mathbf{U}}, \tilde\gamma), \quad \forall\, (\tilde{\mathbf{U}}, \tilde\gamma).
		\end{equation}
		Hence $f_N \equiv f$.
		The identified set of any counterfactual $q$ depends only on the observable $\mathbf{X}$ and the known model $f$; since $f_N = f$, the identified set is unchanged.
	\end{proof}
	
	In other words, normalization WLOG cannot create, destroy, enlarge, or shrink the identified set of any counterfactual $q$ defined on the original space of unknowns.
	
	We should make an important clarification: the theorem concerns a fixed counterfactual $q$, the same function evaluated before and after normalization. Normalization does make a normalization-\emph{relative} object single-valued, but that is a different quantity, namely $q$ composed with the normalization map $\psi_N$, not $q$ itself.
	The theorem says that for the original $q$, imposing normalization changes nothing.
	
	It should also be emphasized that ``without loss of generality'' does not mean ``powerful.'' It means normalization causes \emph{no change} in the model's ability to inform about any counterfactual. Normalization is purely a labeling convenience.

	\begin{remark}[Auxiliary restrictions and end-to-end WLOG]\label{rem:auxiliary}
		Theorem~\ref{thm:wlog} establishes that normalization is WLOG for identification, but this theorem is relative to a given model with all its specifications and assumptions held fixed.  Any object introduced \emph{after} normalization (a different counterfactual, a regularity condition, a metric, a topology, a smoothness requirement, or a prior, imposed beyond the original model) must itself be well-defined on the quotient space for the WLOG interpretation to extend to that object. Formally, let $R \subseteq \cU \times \Gamma$ denote an auxiliary restriction imposed after normalization. If $R$ is a union of modeling equivalence classes, then imposing $R$ preserves the WLOG status.
		If $R$ is not a union of equivalence classes, then $R$ is a substantive restriction: it refines the model, alters the quotient structure, and may change which counterfactuals are admissible or identified.
		A notable applied example is the ``re-normalization'' studied by \citet{agostinelli2025estimating} and \citet{freyberger2025normalizations} in dynamic latent factor models: fixing the scale of a latent skill variable each period is WLOG with fully flexible production technologies, but becomes over-identifying and restrictive when combined with CES production functions, where the relative scale across periods is already pinned down by functional form.
		The discussion in Section~\ref{sec:discussion} develops this point through concrete examples.
	\end{remark}

	%% ================================================================
	\section{Examples}
	\label{sec:examples}
	%% ================================================================
	
	Normalization issues arise across a wide range of structural models. The papers discussed in the literature review provide detailed analyses in several specific settings: \citet*{hamilton2007normalization} study mixture models, structural VARs, and cointegration; \citet{agostinelli2025estimating} and \citet{freyberger2025normalizations} study dynamic latent-factor models of skill formation. Rather than revisit those settings, we illustrate the general framework through three different econometric models that highlight the larger-picture issues: how modeling equivalence classes are determined, how normalization-free and normalization-dependent counterfactuals are distinguished, and how percentage welfare changes can be particularly misleading. For each, we take the model---including all its maintained assumptions---as given, identify the modeling equivalent transformations, characterize the equivalence classes, and classify key counterfactuals. The examples are not meant to evaluate the merits of the underlying models; their purpose is to show that, conditional on whatever model a researcher has adopted, the normalization and its implications for counterfactual analysis should be made explicit.
	
	\subsection{Binary Choice Models}
	\label{subsec:binary}
	
	Consider the standard binary response model $Y_i = \ind\{X_i'\boldsymbol{\beta}_0 \ge \varepsilon_i\}$, where $X_i = (1, X_{i,-1}')' \in \R^d$ includes a constant, $\boldsymbol{\beta}_0 \in \R^d$, and $\varepsilon_i$ has CDF $F_\varepsilon$ that is continuous, strictly increasing, and differentiable with density $f_\varepsilon$.
	For any $a \in \R$ and $b > 0$, define the positive affine transformation
	\begin{equation}\label{eq:binary_trans}
		\tilde\varepsilon_i = a + b\varepsilon_i, \qquad \tilde\beta_{0,1} = a + b\beta_{0,1}, \qquad \tilde{\boldsymbol{\beta}}_{0,-1} = b\boldsymbol{\beta}_{0,-1}.
	\end{equation}
	By construction, $X_i'\tilde{\boldsymbol{\beta}}_0 \ge \tilde\varepsilon_i$ if and only if $X_i'\boldsymbol{\beta}_0 \ge \varepsilon_i$, and the transformed error $\tilde\varepsilon_i$ inherits continuity and strict monotonicity from $\varepsilon_i$.
	Hence \eqref{eq:binary_trans} is a modeling equivalent transformation with two free parameters $(a,b)$, reflecting one location and one scale degree of freedom.
	
	It is important to distinguish the researcher's choice of distributional family, such as probit and logit, from the normalization that operates \emph{within} the chosen family. The choice of family is a substantive modeling decision that determines the shape of $F_\varepsilon$, and this choice is not a normalization in our sense.
	Once the family is chosen, however, the location and scale of $F_\varepsilon$ are not separately identified from $\boldsymbol{\beta}_0$, and it is here that normalization enters: setting $\varepsilon_i \sim \mathcal{N}(0,1)$ or $Logistic(0,1)$. 
	Each such restriction selects one representative from the positive affine equivalence class.
	In the semiparametric case where $F_\varepsilon$ is left unspecified (assumed only continuous and strictly increasing), only location-scale restrictions, such as median-zero plus unit interquartile range, are genuine normalizations.
	
	Several standard quantities of interest are normalization-free.
	Coefficient ratios $\beta_{0,j}/\beta_{0,k}$ for $j,k \ge 2$ are invariant because $\tilde\beta_{0,j}/\tilde\beta_{0,k} = (b\beta_{0,j})/(b\beta_{0,k}) = \beta_{0,j}/\beta_{0,k}$.%
	\footnote{The intercept $j=1$ is excluded because the location shift prevents cancellation: $\tilde\beta_{0,1}/\tilde\beta_{0,k} = (a + b\beta_{0,1})/(b\beta_{0,k})$, which depends on~$a$.}
	Marginal effects are invariant because the density rescaling exactly offsets the coefficient rescaling:
	\begin{equation}\label{eq:binary_me}
		f_{\tilde\varepsilon}(x'\tilde{\boldsymbol{\beta}}_0)\,\tilde\beta_{0,j}
		= \tfrac{1}{b}\,f_\varepsilon\!\bigl(\tfrac{x'\tilde{\boldsymbol{\beta}}_0 - a}{b}\bigr) \cdot b\beta_{0,j}
		= f_\varepsilon(x'\boldsymbol{\beta}_0)\,\beta_{0,j}.
	\end{equation}
	Choice probabilities $\PP(Y_i = 1 \mid X_i = x) = F_\varepsilon(x'\boldsymbol{\beta}_0)$ are invariant because $F_{\tilde\varepsilon}(x'\tilde{\boldsymbol{\beta}}_0) = F_\varepsilon\bigl(\frac{x'\tilde{\boldsymbol{\beta}}_0 - a}{b}\bigr) = F_\varepsilon(x'\boldsymbol{\beta}_0)$. Note that choice probabilities and marginal effects are also directly identified from data, so it is not surprising at all that they are invariant under normalization. 
	
	By contrast, the level of latent utility $x'\boldsymbol{\beta}_0$ is normalization-dependent: under the transformation, $x'\tilde{\boldsymbol{\beta}}_0 = bx'\boldsymbol{\beta}_0 + a$, which depends on both the scale $b$ and the location $a$.
	Average utility changes $\E[\Delta x'\boldsymbol{\beta}_0]$ scale with $b$ (and shift with $a$ when $\Delta x$ includes a change in the intercept).
	Percentage welfare changes are even more problematic.
	Writing $W(x) := x'\boldsymbol{\beta}_0$ for latent welfare under the original parameterization, the percentage change from $x$ to $x'$ is
	\begin{equation}\label{eq:binary_pct}
		\frac{W(x') - W(x)}{W(x)} = \frac{\Delta x'\boldsymbol{\beta}_0}{x'\boldsymbol{\beta}_0}.
	\end{equation}
	Under the positive affine transformation \eqref{eq:binary_trans}, this becomes
	\[
	\frac{\widetilde{W}(x') - \widetilde{W}(x)}{\widetilde{W}(x)} = \frac{b\,\Delta x'\boldsymbol{\beta}_0}{b\,x'\boldsymbol{\beta}_0 + a} \neq \frac{\Delta x'\boldsymbol{\beta}_0}{x'\boldsymbol{\beta}_0}
	\]
	whenever $a \neq 0$: the location shift contaminates the denominator, making the percentage depend on the normalization's zero point, which corresponds exactly to the temperature-change example in the introduction.
	
	In applied binary choice, marginal effects and coefficient ratios are the standard quantities of interest for a good reason: they are the normalization-free counterfactuals. Claims about welfare levels or willingness-to-pay in dollar terms (unless expressed as a ratio) are normalization-dependent and should be interpreted with care.
	
	\subsection{Multinomial Choice and Demand Models}
	\label{subsec:blp}
	
	Multinomial discrete choice models, originating with \citet{mcfadden1974conditional} and extended to differentiated products markets by \citet*{berry1995automobile}, are among the most widely used tools in empirical industrial organization, marketing, transportation, and public policy.
	They underpin merger simulations, antitrust analysis, and welfare evaluations of product introductions and regulatory interventions, settings where the distinction between normalization-free and normalization-dependent counterfactuals has direct policy consequences.
	
	To illustrate, consider a random-coefficient discrete choice model of demand, where consumer $i$'s utility from product $j$ in market $t$ is $u_{ijt} = \delta_{jt} + \nu_{ijt} + \varepsilon_{ijt}$, with mean utility $\delta_{jt} = x_{jt}'\boldsymbol{\beta}_0 - \alpha_0 p_{jt} + \xi_{jt}$, preference heterogeneity $\nu_{ijt}$, and idiosyncratic shock $\varepsilon_{ijt}$.
	Consumer $i$ chooses $\arg\max_j u_{ijt}$, where  $\varepsilon_{ijt} \sim$ Type-I extreme value distribution.
	For any $a \in \R$ and $b > 0$, the positive affine transformation $\tilde{u}_{ijt} = a + bu_{ijt}$ applied to \emph{all} alternatives (including the outside option) preserves the argmax.
	In the model's decomposition this corresponds to $\tilde\delta_{jt} = a + b\delta_{jt}$ (for all $j$, so the outside option's mean utility becomes $\tilde\delta_{0t} = a$ rather than $0$), $\tilde\nu_{ijt} = b\nu_{ijt}$, and $\tilde\varepsilon_{ijt} = b\varepsilon_{ijt}$.
	The standard normalization $\delta_{0t} = 0$ absorbs the location degree of freedom~$a$, and $\varepsilon_{ijt} \sim$ TIEV absorbs the scale~$b$.
	
	As expected, market shares $s_{jt}$, own-price elasticities, and diversion ratios are all functions of choice probabilities and hence normalization-free.
	
	The consumer surplus case is more nuanced and illustrates the framework's value.
	To obtain a closed form, consider the logit special case (no random coefficients: $\nu_{ijt} = 0$).
	Under logit with scale parameter $\mu$ (the standard normalization sets $\mu = 1$), the expected consumer surplus is
	$CS = (\mu/\alpha_0)\log\bigl(\sum_{j=0}^{J} \exp(\delta_{jt}/\mu)\bigr) + C$,
	where the sum runs over all alternatives including the outside option ($\delta_{0t} = 0$), $\alpha_0$ is the marginal utility of income, and $C$ is an unknown constant.
	The CS \emph{level} is normalization-dependent: the constant $C$ absorbs the location shift $a$, and the scale $\mu$ enters the formula.
	However, the money-metric CS \emph{change} from a policy that moves mean utilities from $\delta$ to $\delta'$ is
	\begin{equation}\label{eq:blp_cs}
		\Delta CS = \frac{\mu}{\alpha_0}\Bigl[\log\sum_{j=0}^{J} \exp(\delta'_{jt}/\mu) - \log\sum_{j=0}^{J} \exp(\delta_{jt}/\mu)\Bigr],
	\end{equation}
	in which the unknown constant $C$ cancels.
	Under the transformation $\tilde\mu = b\mu$, $\tilde\alpha_0 = b\alpha_0$, and $\tilde\delta_{jt} = a + b\delta_{jt}$, one can verify that $\widetilde{\Delta CS} = \Delta CS$: the location $a$ cancels in the difference of log-sums, and the scale $b$ cancels between the numerator $\tilde\mu = b\mu$ and the denominator $\tilde\alpha_0 = b\alpha_0$.
	Money-metric CS changes in the standard quasilinear logit model are therefore normalization-free.
	This cancellation relies on the closed-form log-sum structure of logit; in the full random-coefficients model, where $\Delta CS$ involves integrals over the heterogeneity distribution $\nu_{ijt}$, the scale parameter $b$ enters the integrand nonlinearly and the cancellation is no longer guaranteed without additional structure (e.g., conditions ensuring that scale enters the welfare integral only through the price coefficient).
	
	Percentage welfare changes, however, are again \emph{not} normalization-free.
	Since $\Delta CS / CS$ involves the CS level in the denominator, which retains the arbitrary additive constant from the location normalization, the percentage change can again suffer from the same pathology of the temperature-change example in the introduction: the level depends on the zero point, and any ratio or percentage involving that level inherits the arbitrariness.
	
	The framework thus draws a sharp line within the standard demand toolkit: market shares, elasticities, diversion ratios, and money-metric welfare changes are normalization-free; utility levels and percentage welfare changes are not. 
	
	This distinction is consequential for applied work and may be non-obvious at first. For adaptations of the standard multinomial choice or BLP model to specialized contexts, where additional model ingredients, structures, and assumptions are introduced, the boundary between normalization-free and normalization-dependent quantities can shift in subtle and complicated ways, making explicit verification through the equivalence-class criterion all the more important.
	
	\subsection{Dyadic Network Formation Models}
	\label{subsec:network}
	
	Dyadic network formation models have been important tools for the empirical analysis of social and economic networks. The interplay between individual heterogeneity, homophily, and unobserved shocks in the network setting also creates a richer equivalence class structure than in the binary or multinomial choice settings, making the normalization question both more consequential and more subtle.
	
	Consider the following dyadic network formation model \`a la \citet{graham2017econometric}, more specifically in the semiparametric form as considered in \citet{gao2020nonparametric}:
	\begin{equation}\label{eq:network}
		D_{ij} = \ind\bigl\{w(X_i, X_j) + A_i + A_j \ge U_{ij}\bigr\}, \qquad U_{ij} \sim_{\text{iid}} F,
	\end{equation}
	where $w(\cdot)$ is a (possibly nonparametric) homophily effect, $A_i$ are individual fixed effects, and $F$ is an unknown CDF.
	For any $a, b \in \R$ and $c > 0$, the positive affine transformation
	\begin{equation}\label{eq:network_trans}
		\hat{A}_i = cA_i + a, \quad \hat{w}(x_i, x_j) = cw(x_i, x_j) + b, \quad \hat{U}_{ij} = cU_{ij} + 2a + b
	\end{equation}
	preserves the link indicator and all maintained assumptions (symmetry of $w$, the i.i.d.\ structure of $U_{ij}$, independence).
	The equivalence class has \emph{three} degrees of freedom: two location ($a$, $b$) and one scale ($c$). Here we have one more degree of freedom than in binary choice, because the homophily function $w$ and the fixed effects $A_i$ can absorb separate location shifts.
	
	Following \citet{gao2020nonparametric}, one can normalize by setting $w(\bar{x}, \bar{x}) = 0$ (absorbing $b$), $F^{-1}(\alpha) = 0$ (absorbing $a$), and $F^{-1}(1-\alpha) - F^{-1}(\alpha) = 1$ (absorbing $c$) for some $\alpha \in (0, 1/2)$.
	This \emph{two-quantile normalization}, fixing two quantiles of the unobserved distribution rather than its mean and variance, has notable advantages over the standard alternatives.
	Variance normalization ($\text{Var}(U_{ij}) = 1$) requires finite variance, excluding heavy-tailed distributions such as Cauchy; mean-zero normalization ($\E[U_{ij}] = 0$) requires finite first moments.
	By contrast, the two-quantile normalization requires only that $F$ be continuous and strictly increasing on the relevant range, a condition already maintained in the model, and is therefore applicable under minimal regularity.
	This illustrates a broader principle: the choice among WLOG-equivalent normalizations, while irrelevant for identification, can matter for the regularity conditions available to the subsequent estimation and inference program.
	
	Linking probabilities $\PP(D_{ij} = 1 \mid X_i, X_j, A_i, A_j)$ are normalization-free, as is the shape of $w$ (the equivalence class preserves $w$ up to $cw + b$, so curvature and relative patterns survive) and the ranking of fixed effects across individuals (since $\hat{A}_i = cA_i + a$ preserves orderings).
	The absolute levels of $A_i$ and $w(x_i, x_j)$, by contrast, depend on $a$, $b$, and $c$, and similarly percentage changes in $w(x_i, x_j)$ are normalization-dependent.
	
	A useful practice, adopted in \citet{gao2020nonparametric}, is to carry out identification arguments with the normalization constants left as free unknowns and present results as equivalence classes under positive affine transformations.
	This ``gives back'' the normalization and makes transparent which properties of $(w, A, F)$ are genuinely identified and which are artifacts of coordinate choice.
	
	The analysis in \citet{gao2020nonparametric} also illustrates a subtle but important point: while normalization is ``useless'' for identification in the sense of Theorem~\ref{thm:wlog}  (i.e., it cannot create or destroy identification of any counterfactual), a judicious choice of normalization can make the identification \emph{argument} substantially easier.
	In the network setting, conditional linking probabilities are naturally expressed as CDF values, which are in turn linked to quantiles via the inverse transformation.
	The two-quantile normalization $F^{-1}(\alpha) = 0$ and $F^{-1}(1-\alpha) = 1$ provides two ``free starting points'' for the unknown CDF, as known anchor values from which the rest of $F$ can be constructed step by step. These anchors are later ``given back'' when the final result is stated as an equivalence class, but their temporary availability streamlines the constructive identification proof. This analytical convenience is fully consistent with the WLOG theorem: normalization changes no identification \emph{result}, but it can simplify the \emph{path} to that result.
	
	Two extensions of the network model considered in \cite{gao2020nonparametric} further illustrate how richer structure can reduce indeterminacy.
	First, under the parametric specification $w(x_i, x_j) = \tilde{w}(x_i, x_j)'\boldsymbol{\beta}_0$ with $\tilde{w}(x, x) := |x-x| = 0$ by construction, the first location normalization $w(\bar{x}, \bar{x}) = 0$ becomes automatic and no longer consumes a degree of freedom.
	Only the two normalizations on $F$ are needed, reducing the equivalence class from three to two degrees of freedom. Second, if the additive coupling $A_i + A_j$ is replaced by a known nonlinear function $\phi(A_i, A_j)$, the curvature of $\phi$ restricts the class of admissible transformations.
	The positive affine transformation $\hat{A}_i = cA_i + a$ induces $\phi(cA_i + a, cA_j + a)$, which generally differs from $c\phi(A_i, A_j) + \text{(affine terms)}$ unless $\phi$ is additively separable. 
	
	This illustrates the general principle that strengthening model structure shrinks the equivalence classes and makes more counterfactuals normalization-free. Of course, stronger structure comes at a cost: identification becomes more reliant on the model, and concerns about model misspecification become more salient. The framework clarifies this trade-off: richer assumptions buy smaller equivalence classes and more normalization-free counterfactuals, but the validity of those counterfactuals is only as good as the maintained assumptions that define the classes.

	%% ================================================================
	\section{Normalization, Singularity, and Model Modification}
	\label{sec:modification}
	%% ================================================================
	
	Theorem~\ref{thm:wlog} shows that genuine normalization, which selects one representative from each modeling equivalence class, is always WLOG for identification.
	In practice, however, operations that go beyond representative selection are sometimes performed under the same label.
	
	Section~\ref{subsec:augmentation} considers singularities in the target functional: extending a counterfactual beyond the domain where it is naturally well-defined, typically to ``patch'' a boundary singularity, changes the estimand, not just the coordinates.
	Section~\ref{sec:sphere} considers singularities in the normalized parameter space: the regularity conditions (topological, metric, and compactness assumptions) imposed for estimation and inference are stated on the normalized space, not on the equivalence classes.
	Because these conditions are properties of the normalization rather than the equivalence classes, different normalizations that are equally valid for identification can lead to different, and sometimes contradictory, asymptotic conclusions.
	
	The two cases are, in a sense, flip sides of the same singularity story: domain augmentation attempts to repair a pre-existing singularity in the target functional, while certain normalizations \emph{create} a new singularity in the parameter space.
	In both cases, the framework provides a clean diagnostic: the operation goes beyond genuine normalization whenever it fails to preserve the equivalence-class structure.
	
	\subsection{Singularity in the Target Functional}
	\label{subsec:augmentation}
	
	A \emph{boundary singularity} arises when a counterfactual $q_0$ is normalization-free on a ``regular'' subset of the parameter space but cannot be coherently extended to the boundary.
	When the target functional is only partially defined, well-defined on a regular domain but singular at the boundary, any attempt to extend it to the full space faces a tension among three desiderata: \emph{fidelity} (the extension agrees with the original on the regular domain), \emph{invariance} (the extension is $Q_\sim$-measurable on the full space), and \emph{regularity} (the extension is finite-valued and continuous at the boundary).
	The following result shows this tension is unavoidable whenever the quotient-space functional has no well-behaved limit at the boundary.
	
	\begin{thm}[Boundary Extension Trilemma]\label{thm:trilemma}
		Let $\Theta_0 \subsetneq \cU \times \Gamma$ be an open invariant subset (a union of modeling equivalence classes). Let $q_0 : \Theta_0 \to \R$ be continuous and $Q_\sim$-measurable, so that   $q_0 = \bar{q}_0 \circ \pi$ for a map $\bar{q}_0 : \pi(\Theta_0) \to \R$, where $\pi$ is the canonical projection in \eqref{eq:factorization}.
		If $\bar{q}_0$ does not admit a \textbf{finite continuous extension} to some $p^* \in \overline{\pi(\Theta_0)} \setminus \pi(\Theta_0)$, then \textbf{no} map $\tilde{q} : \cU \times \Gamma \to \R$ can simultaneously satisfy:
		\begin{enumerate}[label=(\alph*), nosep]
			\item \emph{Fidelity}: $\tilde{q}\big|_{\Theta_0} = q_0$.
			\item \emph{Invariance}: $\tilde{q}$ is $Q_\sim$-measurable on $\cU \times \Gamma$.
			\item \emph{Regularity}: the induced map on the quotient $\bar{\tilde{q}}$ is finite and continuous at $p^*$.
		\end{enumerate}
	\end{thm}
	
	\begin{proof}
		If $\tilde{q}$ satisfies~(b), it factors as $\tilde{q} = \bar{\tilde{q}} \circ \pi$ for some $\bar{\tilde{q}} : Q_\sim \to \R$.
		By~(a), $\bar{\tilde{q}} = \bar{q}_0$ on $\pi(\Theta_0)$.
		By~(c), $\bar{\tilde{q}}$ is finite and continuous at $p^*$.
		Hence $\bar{\tilde{q}}$ is a finite continuous extension of $\bar{q}_0$ to $p^*$, contradicting the hypothesis.
	\end{proof}
	
	A practical sufficient condition for the hypothesis is that there exist sequences $p_n, r_n \in \pi(\Theta_0)$ with $p_n \to p^*$ and $r_n \to p^*$ but $\lim \bar{q}_0(p_n) \neq \lim \bar{q}_0(r_n)$ (including the case where one limit diverges), which we call the \emph{non-unique limit test} on the quotient. One response to the trilemma is to shrink the domain to exclude $p^*$, sacrificing generality rather than any of the three properties.
	Augmentation takes the opposite approach: it extends the functional to the boundary, but typically sacrifices invariance~(b) and may also sacrifice exact fidelity~(a), producing a well-defined number for a different, normalization-dependent object.
	Proposition~\ref{prop:fixed_point} below identifies one algebraic mechanism behind the singularity: a fixed point of the transformation family that generates the non-extendability condition.
	\medskip
	
	The ``logs with zeros'' problem studied by \citet{chen2024logs} provides a vivid illustration.
	The log transformation $\log(Y)$, commonly used to express percentage effects, is well-defined for strictly positive $Y$ and yields a scale-invariant ATE (since $\log(aY_1) - \log(aY_0) = \log Y_1 - \log Y_0$), but has a singularity at $Y = 0$.
	Any extension, such as choosing $c$ in $\log(c+Y)$, or switching to $\operatorname{arcsinh}(Y)$, attempts to repair this singularity, but imports a new functional that does not coincide with $\log Y$ on the positive reals and whose induced ATE is not scale-invariant.
	The following result provides a complementary viewpoint of the algebraic mechanism behind the singularity.
	
	\begin{prop}[Fixed-Point Singularity]\label{prop:fixed_point}
		Let $\mathcal{G}$ be a collection of transformations on a space $\mathcal{Y}$, let $\rho : \mathcal{G} \to \R$ assign a real number to each transformation, and let $m : \mathcal{Y} \setminus \{p\} \to \R$ satisfy
		\begin{equation}\label{eq:equivariance}
			m(g \cdot y) = \rho(g) + m(y), \qquad \forall\, g \in \mathcal{G},\ y \in \mathcal{Y} \setminus \{p\}.
		\end{equation}
		Suppose $\rho(g) \neq 0$ for some $g \in \mathcal{G}$, and that $p \in \mathcal{Y}$ is a fixed point in the sense of $g \cdot p = p$ for all $g \in \mathcal{G}$.
		Then:
		\begin{enumerate}[label=(\roman*), nosep]
			\item $m$ cannot be extended to $p$: no value $m(p) \in \R$ is consistent with \eqref{eq:equivariance}.
			\item For any extension $\tilde{m} : \mathcal{Y} \to \R$ that agrees with $m$ on $\mathcal{Y} \setminus \{p\}$, the induced difference $\tilde{q}(y_1, y_0) := \tilde{m}(y_1) - \tilde{m}(y_0)$ is not invariant under $\mathcal{G}$ whenever one argument equals $p$.
		\end{enumerate}
	\end{prop}
	
	\begin{proof}
		For~(i): if $m(p)$ were defined, then $m(p) = m(g \cdot p) = \rho(g) + m(p)$ for all $g \in \mathcal{G}$, giving $\rho(g) = 0$ for all $g$, contradicting the assumption that $\rho(g) \neq 0$ for some $g$.
		For~(ii): set $y_0 = p$ and let $g \in \mathcal{G}$ satisfy $\rho(g) \neq 0$.
		Then $\tilde{q}(g \cdot y_1, g \cdot p) = \tilde{q}(g \cdot y_1, p) = \tilde{m}(g \cdot y_1) - \tilde{m}(p) = [\rho(g) + \tilde{m}(y_1)] - \tilde{m}(p) = \rho(g) + \tilde{q}(y_1, p) \neq \tilde{q}(y_1, p)$.
	\end{proof}
	
	\noindent In the setting of \citet{chen2024logs}, $\mathcal{G}$ is the collection of scaling transformations $y \mapsto ay$ for $a > 0$ on $\mathcal{Y} = [0, \infty)$, the fixed point is $p = 0$, and $m = \log$ satisfies \eqref{eq:equivariance} on $(0,\infty)$ with $\rho(a) = \log a$.
	Proposition~\ref{prop:fixed_point}(i) says that no value can be assigned to $\log(0)$ consistently with scale equivariance.
	Part~(ii) says that even in the best case, where $\tilde{m}$ agrees exactly with $\log$ on $(0, \infty)$ and assigns an arbitrary finite value $\tilde{m}(0) = c$, the induced ATE $\E[\tilde{m}(Y_1) - \tilde{m}(Y_0)]$ is not scale-invariant whenever the treatment moves some individuals onto or off the boundary $Y = 0$. The commonly used ``log-like'' transformations such as $\log(1 + y)$ and $\operatorname{arcsinh}(y)$ do not even satisfy this best case: they differ from $\log y$ on the entire positive real line, so they sacrifice fidelity in addition to invariance. This is the further critique made in \citet{chen2024logs}, which our Proposition \ref{prop:fixed_point} does not directly capture. 
	
	That said, our extension trilemma (Theorem~\ref{thm:trilemma}) approaches the issue from a complementary perspective based on our broader framework. 
	In terms of Theorem~\ref{thm:trilemma}, the quotient functional $\bar{q}_0$ (the ATE restricted to strictly positive outcomes) is well-defined on the interior of the quotient but does not extend continuously to the boundary class where $P(Y = 0) > 0$: approximating the atom at $0$ by mass at $\delta > 0$ contributes $P(Y_0 = 0) \cdot \log(a\delta)$ to the ATE of the $a$-scaled element, which shifts by $P(Y_0 = 0) \cdot \log a$ across elements of the same boundary class, so the non-unique limit test is satisfied.
	Any completion therefore achieves regularity~(c) at the boundary but necessarily sacrifices fidelity~(a), invariance~(b), or both.
	
	The structural models of Section~\ref{sec:examples} further illustrate the broader relevance of Theorem \ref{thm:trilemma} beyond Proposition \ref{prop:fixed_point} and fixed-point singularities.  Percentage welfare changes are normalization-dependent through a different mechanism: an \emph{interior} failure rather than a boundary failure.
	The positive affine family $y \mapsto a + by$ has no fixed point, so Proposition~\ref{prop:fixed_point} does not apply.
	Instead, the location shift~$a$ enters the utility \emph{level}, making the denominator of any percentage normalization-dependent. The percentage is not $Q_\sim$-measurable even on its natural domain: it varies within equivalence classes, so Lemma~\ref{lem:limit} already implies it cannot be point-identified.
	The two settings thus illustrate complementary challenges for identification and normalization: Theorem~\ref{thm:trilemma} governs the \emph{boundary} case (a normalization-free functional that resists extension), while Lemma~\ref{lem:limit} governs the \emph{interior} case.
	
	The takeaway from this subsection is that if a purported ``normalization'' is needed to make the target functional well-defined, the problem lies in the definition of the target, or in the model, not in the parameterization.
	This warning is not merely abstract.
	In dynamic latent-factor models of skill formation, \citet{agostinelli2025estimating} show that period-by-period re-normalization can be over-identifying and can bias estimation when the production technology already has known location and scale, while \citet{freyberger2025normalizations} shows more broadly that seemingly innocuous scale and location restrictions can constrain identified parameters and alter counterfactuals in a model-dependent way.
	In the language of the present paper, these are precisely cases in which an operation marketed as a normalization is actually doing substantive modeling work.
	%% ================================================================
	\subsection{Singularity in the Normalized Parameter Space}\label{sec:sphere}
	%% ================================================================
	
	Domain augmentation confronts a pre-existing singularity in the \emph{target functional}. In addition, there are cases where the singularity is not inherited but \emph{introduced} by the ``normalization'', for example, dividing by a designated coefficient creates a coordinate singularity, distorting the geometry of the parameter space. The induced singularity may interact with further model modification in the form of regularity conditions (topological, metric, smoothness) when they are imposed for estimation and inference purposes. These conditions are substantive restrictions typically stated on the normalized parameter space rather than on the equivalence classes, and they need not be preserved across different normalizations of the same classes.
	Two normalizations that are both valid for identification purposes can therefore induce sharply different parameter-space geometries, with direct consequences for the validity of standard asymptotic arguments.
	This point is closely related to \citet*{hamilton2007normalization}, who show in likelihood-based models that poor normalizations can create multimodal posteriors, disjoint confidence regions, and misleading descriptions of statistical uncertainty.
	Our contribution here is complementary: even when two normalizations are equally WLOG for identification, a coordinate chart can still distort the topology and metric imposed downstream on the quotient space.
	
	We illustrate this using the semiparametric binary choice model \`a la \cite{manski1975}. Specifically, consider
	$$Y_i = \ind\{X_i'\boldsymbol{\beta}_0 \ge \varepsilon_i\},\quad \boldsymbol{\beta}_0 \in \R^D \setminus \{0\},$$
	which has scale indeterminacy: each equivalence class is a ray $\{c\boldsymbol{\beta}_0 : c > 0\}$ in $\R^D \setminus \{0\}$, and any normalization that selects one representative per ray is WLOG for identification by Theorem~\ref{thm:wlog}. In the econometric literature, there are two commonly used approaches for scale normalization: the \emph{coordinate normalization} and the \emph{sphere normalization}, which we describe below.
	
	\subsubsection{Special-Coordinate versus Sphere Normalization}
	
	The most commonly adopted approach in the semiparametric binary choice literature, or what we call \emph{special-coordinate normalization} here, normalizes the coefficient on a ``special regressor'': assume $\beta_{0,1} \neq 0$ in addition and set $|\beta_{0,1}| = 1$, so that the normalized parameter space becomes 
	\begin{equation}\label{eq:coord_norm}
		\tilde{\boldsymbol{\beta}}_0 := \Bigl(\operatorname{sgn}(\beta_{0,1}),\; \frac{\beta_{0,2}}{|\beta_{0,1}|},\; \ldots,\; \frac{\beta_{0,D}}{|\beta_{0,1}|}\Bigr)' \in \tilde{B} := \{1, -1\} \times \R^{D-1}.
	\end{equation}
	The normalized parameter space $\tilde{B}$ consists of two copies of $\R^{D-1}$ (one for each sign of $\beta_{0,1}$), and the natural metric is the Euclidean norm on $\R^{D-1}$ applied to the remaining coordinates, with the sign handled separately.
	It is better understood as a pair of branchwise affine charts on the open set $\{\beta_{0,1} \neq 0\}$ than as a globally innocuous normalization of the quotient space. If the sign of $\beta_{0,1}$ is separately identified, one may work on a single branch such as $\beta_{0,1}=1$ or $\beta_{0,1}=-1$; if not, one can keep both branches through $|\beta_{0,1}|=1$. Either way, the deeper issue is unchanged: the chart becomes singular as $\beta_{0,1} \to 0$.
	This approach has been widely adopted since \citet{horowitz1992} in the literature on semiparametric binary and discrete choice models \citep[see also, e.g., ][]{lewbel2000semiparametric}, because it effectively reduces the parameter space to $\R^{D-1}$, on which standard asymptotic theory applies without the need to deal with any spherical geometry.
	
	An alternative normalization approach, which was used in \cite{manski1975,manski1985} and is referred to as \emph{sphere normalization} here, is to project the parameter space onto the unit sphere:
	\[
	\bar{\boldsymbol{\beta}}_0 := \frac{\boldsymbol{\beta}_0}{\norm{\boldsymbol{\beta}_0}} \in \mathbb{S}^{D-1} := \{v \in \R^D : \norm{v} = 1\}.
	\]
	This normalization was adopted in more recent work by \citet{gaoli2026} in a panel multinomial choice setting.
	The relevant normalized parameter space is the sphere itself, equipped with the \emph{great-circle (angular) metric}:
	\begin{equation}\label{eq:gc_metric}
		\rho_{GC}(\bar{\boldsymbol{\beta}}, \underline{\boldsymbol{\beta}}) := \arccos\bigl(\bar{\boldsymbol{\beta}}'\underline{\boldsymbol{\beta}}\bigr) \in [0, \pi],
	\end{equation}
	which measures the angle between two unit vectors and reflects the intrinsic geometry of the quotient space of rays.
	The great-circle metric and the Euclidean metric on $\mathbb{S}^{D-1}$ are \emph{strongly equivalent}: for all $\bar{\boldsymbol{\beta}}, \underline{\boldsymbol{\beta}} \in \mathbb{S}^{D-1}$,
	\begin{equation}\label{eq:strong_equiv}
		\norm{\bar{\boldsymbol{\beta}} - \underline{\boldsymbol{\beta}}} \le \rho_{GC}(\bar{\boldsymbol{\beta}}, \underline{\boldsymbol{\beta}}) \le \frac{\pi}{2}\norm{\bar{\boldsymbol{\beta}} - \underline{\boldsymbol{\beta}}},
	\end{equation}
	so the two metrics generate the same topology, the same notions of convergence, and the same uniform continuity.
	
	\subsubsection{Implications for Estimation and Inference}
	
	Both normalizations are WLOG for identification under their stated assumptions, but some of their consequences diverge, especially once one imposes additional regularity conditions needed for estimation and inference. Four differences are particularly consequential.
	
	\paragraph{Restrictiveness.}
	The sphere $\mathbb{S}^{D-1}$ represents all directions in $\R^D \setminus \{0\}$, while the special-coordinate normalization excludes every $\boldsymbol{\beta}_0$ with $\beta_{0,1} = 0$. Even when $\beta_{0,1} \neq 0$ is assumed, the chart becomes singular as $\beta_{0,1} \to 0$: the excluded hyperplane is a set of measure zero, but its influence extends to an entire neighborhood through the divergence of the normalized coordinates, as illustrated below.
	
	\paragraph{Topological Structures.}
	The two normalizations induce fundamentally different topologies on their parameter spaces.
	The coordinate-normalized space $\tilde{B} = \{1,-1\} \times \R^{D-1}$ is \emph{disconnected} (the two sign components are disjoint) and \emph{unbounded} (the coordinates $\beta_{0,k}/|\beta_{0,1}|$ diverge as the true parameter approaches the excluded hyperplane $\beta_{0,1} = 0$).
	Since compactness requires boundedness, the coordinate normalization destroys it.
	To restore it, as typically required for consistency of extremum estimators, one imposes an ad hoc compactness restriction  $\tilde{\boldsymbol{\beta}}_0 \in \{1, -1\} \times [-M, M]^{D-1}$.
	However, this excludes not only directions with $\beta_{0,1} = 0$ but an entire neighborhood with $|\beta_{0,1}| < 1/M$, a substantive restriction masquerading as a regularity condition.
	By contrast, the sphere $(\mathbb{S}^{D-1}, \rho_{GC})$ is \emph{connected} (for $D \ge 2$), \emph{bounded}, and \emph{closed}, and hence compact by construction, requiring no auxiliary trimming or boundedness assumptions.
	
	\paragraph{Convergence Properties.}
	Most importantly, the coordinate singularity at $\beta_{0,1} = 0$ can reverse convergence verdicts. We give two illustrations with $D = 2$.
	
	\emph{Across sign components.}
	Consider $\tilde{\boldsymbol{\beta}}^{(M)} = (1, M)'$ and $\tilde{\boldsymbol{\gamma}}^{(M)} = (-1, M)'$ in $\tilde{B}$, approaching the excluded hyperplane $\beta_{0,1} = 0$ from opposite sign components.
	Because $\tilde{\boldsymbol{\beta}}^{(M)}$ and $\tilde{\boldsymbol{\gamma}}^{(M)}$ lie in different connected components of $\tilde{B} = \{1,-1\} \times \R^{D-1}$, the chart topology treats them as permanently separated: no sequence in one component can converge to a point in the other.
	Yet both represent rays that approach the same direction $(0,1)'$ as $M \to \infty$.
	On the sphere, let $\bar{\boldsymbol{\beta}}^{(M)} := \tilde{\boldsymbol{\beta}}^{(M)}/\norm{\tilde{\boldsymbol{\beta}}^{(M)}}$ and $\underline{\boldsymbol{\beta}}^{(M)} := \tilde{\boldsymbol{\gamma}}^{(M)}/\norm{\tilde{\boldsymbol{\gamma}}^{(M)}}$.
	Then $\rho_{GC}(\bar{\boldsymbol{\beta}}^{(M)}, \underline{\boldsymbol{\beta}}^{(M)}) \to 0$ as $M \to \infty$: the two directions converge on the sphere, but the disconnectedness of the chart makes this invisible.
	
	\emph{Within a sign component.}
	Now consider $\tilde{\boldsymbol{\beta}}^{(M)} = (1, M)'$ and $\tilde{\boldsymbol{\gamma}}^{(M)} = (1, 2M)'$, both in the $\{1\} \times \R$ component.
	Their coordinate distance is $\rho_{\mathrm{Euc}}(\tilde{\boldsymbol{\beta}}^{(M)}, \tilde{\boldsymbol{\gamma}}^{(M)}) = M \to \infty$, suggesting divergence.
	Yet both rays converge to $(0,1)'$ on the sphere, and the great-circle distance between their projections tends to zero: the chart inflates distances near the singularity without bound.
	
	In both cases, what is nonconvergent under the special-coordinate normalization is convergent under the sphere normalization.
	Any asymptotic result formulated in $(\tilde{B}, \rho_{\mathrm{Euc}})$ carries an implicit dependence on the coordinate chosen for normalization, a form of nonuniformity that the sphere normalization avoids entirely.
	
	\paragraph{Metric Structures.}
	The metric distortion documented above has direct consequences for inference. Standard errors and confidence intervals inherit the Euclidean metric of the normalized coordinate space. Under the special-coordinate normalization, a 95\% confidence interval for a parameter is an interval in $\R^{D-1}$. However, the Euclidean metric on $\tilde{B}$ is not the intrinsic metric on the space of equivalence classes (rays from the origin). A given length in $\tilde{B} = \{1,-1\} \times \R^{D-1}$ can correspond to a large arc near the poles ($|\tilde{\boldsymbol{\beta}}_{0,-1}|$ small) but a vanishingly small arc near the equator ($|\tilde{\boldsymbol{\beta}}_{0,-1}| \to \infty$).
	Consequently, the magnitude of the standard error and the ``width'' of a confidence interval reported in coordinate-normalized units are chart-dependent and not uniformly comparable to the intrinsic quotient geometry, especially near the singularity. By contrast, on the sphere $(\mathbb{S}^{D-1}, \rho_{GC})$, the strong metric equivalence~\eqref{eq:strong_equiv} guarantees that Euclidean and great-circle distances are uniformly comparable, so standard errors retain their geometric meaning across the entire parameter space.
	
	\subsubsection{Summary}
	The binary choice example illustrates a more general lesson: a normalization that introduces a coordinate singularity can have consequences for estimation and inference that go far beyond the choice of representative, because the regularity conditions that underpin asymptotic theory (topological structure, metric structure, uniformity, and smoothness) are imposed on the \emph{normalized parameter space}, not on the equivalence classes themselves.
	
	This is not merely a theoretical curiosity. In any model with scale, location-scale, or more complex indeterminacy, the choice of normalization determines the topology in which estimators converge, the compactness conditions available to the analyst, and the uniformity properties of the resulting inference.
	Together with Section~\ref{subsec:augmentation}, the two subsections illustrate the two faces of the singularity problem that normalization can encounter. Domain augmentation attempts to repair a pre-existing singularity: $\log(0)$ is undefined, and any completion that extends the functional to the boundary necessarily sacrifices invariance or fidelity. The coordinate normalization, by contrast, \emph{creates} a singularity: dividing by $\beta_{0,1}$ maps the excluded hyperplane $\{\beta_{0,1} = 0\}$ to infinity, so that the normalized coordinates $\beta_{0,k}/|\beta_{0,1}|$ diverge as $\beta_{0,1} \to 0$, distorting the topology and metric of the parameter space.
	The framework of this paper, and the concept of end-to-end WLOG normalization introduced in Section~\ref{subsec:end_to_end}, provide the language for diagnosing when a normalization that is WLOG for identification may nonetheless fail to be WLOG for the full statistical analysis.
	
	%% ================================================================
	\section{Discussion}
	\label{sec:discussion}
	%% ================================================================
	
	\label{subsec:end_to_end}
	
	The preceding sections show that singularities can arise in two complementary ways: a pre-existing singularity in the target functional may resist extension to the boundary (Section~\ref{subsec:augmentation}), or a normalization may introduce a coordinate singularity that distorts the geometry of the parameter space (Section~\ref{sec:sphere}).
	
	This motivates what we call \emph{end-to-end WLOG normalization}: assessing normalization not in isolation at the identification stage, but jointly with the target counterfactual, the regularity conditions, and any auxiliary structure used downstream. A normalization is end-to-end WLOG only if every downstream object, whether a counterfactual parameter or a regularity condition for estimation and inference, is coherent with respect to the relevant quotient space. When verifying this for every downstream object is impractical, the researcher should at minimum select one specific object of primary interest and check whether the chosen normalization is end-to-end WLOG with respect to that object.
	
	Seen through this lens, the papers discussed in the literature review line up naturally.
	\citet*{hamilton2007normalization} emphasize the downstream inferential geometry of normalization, while \citet{agostinelli2025estimating} and \citet{freyberger2025normalizations} emphasize that some restrictions advertised as normalizations are in fact over-identifying or misspecifying assumptions once combined with richer model structure.
	The present framework unifies these insights: a proper normalization is a representative choice on a fixed equivalence class, whereas an operation that changes the target counterfactual, shrinks the class, or distorts the downstream geometry belongs to modeling, not normalization.
	
	The results extend naturally to partial identification: the identified set is a union of equivalence classes, and the WLOG Theorem applies equally.
	More broadly, richer model structure shrinks the equivalence classes and makes more counterfactuals normalization-free, a useful design principle for model specification.

	To reiterate the scope of the paper: we do not prescribe which models or maintained assumptions a researcher should adopt. A researcher working on a specific structural model is free to make whatever assumptions the application demands, and must justify those assumptions to their own audience. However, once the substantive assumptions are in place, the normalization issue should be explicitly addressed: what equivalence classes does the model induce, and how does each target counterfactual---as well as any regularity condition imposed for estimation and inference---interact with those classes? Making this accounting explicit is the paper's central recommendation.

	Concretely, the framework suggests the following checklist:
	\begin{enumerate}[nosep]
		\item What are the modeling equivalent transformations, and how many degrees of freedom does each equivalence class have?
		\item Is each target counterfactual constant on the relevant equivalence classes?
		\item If not, is the number interpreted explicitly as normalization-relative?
	\end{enumerate}
	Three practical recommendations follow.
	First, report the equivalence class: if the main result is ``point identification under normalization,'' the reader should be told what the unnormalized class looks like.
	Second, separate identification from parameterization: the choice of normalization may be irrelevant for economic content yet matter for estimation, inference, and computation (Section~\ref{sec:sphere}).
	Third, favor normalization-free counterfactuals: policy questions stated in terms of choice probabilities, entry probabilities, or network statistics are more likely to be invariant and hence structurally identified.
	This last recommendation aligns with the findings of \citet{freyberger2025normalizations}.
	%% ================================================================
	\section{Conclusion}
	\label{sec:conclusion}
	%% ================================================================
	
	We have presented a general framework for analyzing the interplay between identification, normalization, and counterfactual parameters in econometric models. The framework is deliberately model-agnostic: it does not evaluate the merits of particular maintained assumptions, but rather provides tools that any researcher can apply, conditional on their chosen model, to diagnose whether a normalization is genuinely without loss and whether it quietly changes the object of interest. The concrete illustrations across binary choice, demand estimation, and network formation demonstrate how this diagnostic works in practice.
	
	The singularity perspective developed in Section~\ref{sec:modification} suggests that the consequences of normalization extend well beyond identification. A normalization that is innocuous for point identification may nonetheless introduce coordinate singularities that distort the topology and metric of the parameter space, with implications for the validity of regularity conditions, consistency of estimators, and uniformity of inference. A systematic study of which normalizations are ``safe'' for downstream statistical analysis, and whether singularity-free alternatives exist in broader classes of models, is an important direction for future work.

	%% ================================================================
	\bibliographystyle{ecta}
	\bibliography{IDNorm}
	%% ================================================================
	
\end{document}